\documentclass[12pt]{article}
\usepackage[T1]{fontenc}
\usepackage[utf8]{inputenc}
\usepackage[english]{babel}
\usepackage[margin=2.5cm,top=2.5cm]{geometry}
\usepackage{amsmath,amsthm,amsfonts,amssymb}
\usepackage{fancybox}
\usepackage{bbm}
\usepackage{mathrsfs}
\usepackage{tikz}
\usepackage{amsmath,textcomp,amssymb,geometry,graphicx,enumerate}

\usepackage[colorinlistoftodos]{todonotes}
\usepackage{forest}
\usepackage{algorithm} 
\usepackage[noend]{algpseudocode} 
\usepackage{hyperref}
\usetikzlibrary{calc,positioning}
\usepackage{tikz}
\usepackage{fontawesome}
\definecolor{myyellow}{RGB}{255, 255, 0}
\definecolor{LightGreen}{HTML}{CCFFCC}
\definecolor{LightBlue}{HTML}{E6F3FF}
\usepackage{xcolor}
\usepackage{caption}
\usetikzlibrary{positioning}
\usepackage{graphicx}
\usepackage{comment}
\usepackage{listings}
\usepackage{tikz-cd}

\usepackage{mdframed}
\usepackage{enumitem}   
\usepackage{multirow}
\usetikzlibrary{arrows,positioning, calc}
\tikzstyle{vertex}=[draw,fill=black!15,circle,minimum size=18pt,inner sep=0pt]
\hypersetup{
    colorlinks=true,
    linkcolor=blue,
    filecolor=magenta,      
    urlcolor=blue,
}
\usepackage[english]{babel}
\newtheorem{theorem}{Theorem}
\usepackage{upgreek}
\newtheorem{corollary}{Corollary}[theorem]
\newtheorem{observation}{Observation}
\newtheorem{lemma}[theorem]{Lemma}
\usepackage{verbatim}
\usepackage{graphicx}

\newcommand{\N}{\mathbb{N}}
\newcommand{\V}{\mathcal{V}}
\newcommand{\dichi}{\overrightarrow{\chi}}
\def\c{\mathrm{\raisebox{.07cm}{$\chi$}}}

\title{Matrix Rationalization via Partial Orders}
\author{
    Agnès Totschnig\thanks{McGill University. \texttt{agnes.totschnig@mail.mcgill.ca}, supported by FRQNT Grant 332481.} \and 
    Rohit Vasishta\thanks{McGill University. \texttt{rohit.vasishta@mail.mcgill.ca}, supported by a McCall MacBain Finalist Award.} \and 
    Adrian Vetta\thanks{McGill University. \texttt{adrian.vetta@mcgill.ca}, supported by NSERC Discovery Grant 2022-04191.}
}
\date{}

\begin{document}

\maketitle


%
\begin{abstract}
A {\em preference matrix} $M$ has an entry for each pair of candidates 
in an election whose value $p_{ij}$ represents the proportion of voters that prefer candidate $i$ over candidate $j$.
The matrix is {\em rationalizable} if it is consistent with a set of voters whose preferences are total orders.
A celebrated open problem asks for a concise characterization of rationalizable preference matrices.
In this paper, we generalize this matrix rationalizability question and study when a 
preference matrix is consistent with a set of voters whose 
preferences are partial orders of width~$\alpha$.
The width (the maximum cardinality of an antichain) of the partial order is a natural measure of the rationality of a voter; indeed,
a partial order of width $1$ is a total order.
Our primary focus concerns the {\em rationality number}, the minimum width required to rationalize a preference matrix.
We present two main results. The first concerns the class of half-integral preference matrices, where we show the key parameter required in evaluating the rationality number is the chromatic number of the undirected {\em unanimity graph} associated with the preference matrix $M$.
The second concerns the class of integral preference matrices, where we show the key parameter now is the dichromatic number of the directed {\em voting graph} associated with~$M$.
\end{abstract}

\section{Introduction}

At the heart of macroeconomics is the concept of {\em choice}~\cite{BCD16}. For example, what should a consumer (respectively, producer) demand (respectively, supply) given market prices?
More generally, given a set of options to choose from an agent selects an option(s). The agent is considered {\em rational} if it always selects the best choice among the given options. 
Specifically, a rational agent has a {\em total order} over the entire collection of options and, presented with a subset of the options, always
chooses the option highest in the ordering.

\subsection{Rationalizable Choice Data}
But how can we evaluate whether or not agents are rational?
The simple answer is to test the data. Is a collection of observational
choice data consistent with decision-making by a group of rational agents?
A classical way to model this is via an election. 
Assume there are $n$ candidates over which we have pairwise choice data;
that is $p_{ij}$ is the proportion of voters that prefer candidate $i$ over candidate $j$. This induces a non-negative {\em preference matrix} 
$M = (p_{ij})_{i,j \in [n]}$, where $p_{ii}=0$ for each candidate $i$ and
$p_{ij}+p_{ji}=1$ for every pair of candidates $i$ and $j$.
Is this preference matrix compatible with an electorate of rational voters,
where each voter ranks the candidates via a total order?\\[0.1cm] 
\noindent{\sc The Matrix Rationalizability Problem:} Given a preference matrix $M$, does there exist a set of rational voters such that, for any pair of candidates, a random voter prefers $i$ over $j$ with probability exactly $p_{ij}$?\footnote{For irrational matrices one may ask if the matrix is compatible with a probability distribution over total orders.}\\[0.1cm]
\indent The fundamental value of this problem is highlighted by its importance in a wide range of disciplines. 
The problem and its variants have been studied in depth in mathematical psychology (primarily with respect to human decision making)~\cite{Fal78,Menezes,Suck16},
economics (w.r.t. econometrics and behavioural economics)~\cite{Houthakker50,Mar60,MR11,DP07},
social choice (w.r.t. voting theory)~\cite{MONTES2020},
operations research (w.r.t. consumer choice and advertising)~\cite{MR11,CCP21},
combinatorial optimization (w.r.t. geometry and integrality in mathematical programming)~\cite{Meg77,GJR83,Fishburn87}, and 
theoretical computer science (w.r.t. algorithms and computational complexity)~\cite{Bachmeier19}. 

Unfortunately, the problem is extremely difficult. A good characterization of rationalizable preference matrices has eluded the research community for over 60 years.
Consequently, the objectives of this paper are more modest, but more general.
Rather than study preference matrices that are compatible with 
a collection of total orders, we study preference matrices that 
are compatible with a collection of {\em partial orders}.
There are two major advantages to this approach. First, from a practical
perspective, the problem of non-existence is avoided. Every preference matrix is compatible with a set of voters with partial order preferences (see Observation~\ref{obs:exist}). Second, and more substantially, the use of partial
orders induces a natural measure of approximate rationality.
Specifically, a total order corresponds to a poset of width $1$, where the
width is the maximum cardinality of an antichain in the poset.
More generally, the smaller the width of a partial order the closer it is
to a total order. 
Intuitively, the smaller the width the more 
``decisive'' and rational the voter, in that it has strong preferences over larger sets of candidates. In contrast, the poset of a voter with higher width has a higher number of linear extensions; the voter is thus more ambiguous and less decisive. Consequently, we say that a voter whose preferences are given by a partial order of width at most $\alpha$ is $\alpha$-rational. Thus, a $1$-rational voter is rational. Further, we say that a preference matrix $M$ is 
$\alpha$-rationalizable if it can be explained by a set of voters who are $\alpha$-rational. Let's begin by formalizing exactly what this means
with examples.

\subsection{The Model}
A (strict) {\em partial order} $\succ$ over a set $[n]=\{1,2,\dots, n\}$ of candidates satisfies, for any $i,j,k \in [n]$, the following three properties:
\\
\\
\indent{\tt Irreflexivity:} {\sc not} $i\succ i$. \\
\indent{\tt Asymmetry:} If $i\succ j$, then {\sc not} $j\succ i$. \\
\indent{\tt Transitivity:} If $i\succ j$ and $j\succ k$, then $i\succ k$.\\ \\
We say that a pair of candidates $i$ and $j$ are {\em comparable} if either $i\succ j$ or $j\succ i$; else they are {\em incomparable}.
A partial order is a {\em total order} if every pair of candidates is comparable.

We assume each voter $v$ has a personal set of preferences given by a partial order $\succ_v$ over the candidates, where the voter {\em strongly prefers} candidate $i$ over $j$ if $i\succ_v j$ (we omit the subscript and write $i\succ j$ if the context is clear). We say the voter {\em weakly prefers} 
$i$ over $j$ if either $i\succ j$ (strongly prefers) or if $i$ and~$j$ are incomparable (voter is indifferent).
Recall that a partial order $\succ$ induces a {\em poset} $\mathcal{P}$ over the candidates.
A {\em chain} is a subset of candidates in $\mathcal{P}$ that induces a total order. An {\em antichain} in $\mathcal{P}$ is a subset of pairwise incomparable candidates. 
We say that the voter is {\em $\alpha$-rational} if the maximum cardinality of an antichain in its partial order is at most $\alpha$.

A preference matrix is a non-negative matrix $M=(p_{ij})$, where $p_{ii}=0$ and $p_{ij}+p_{ji}=1$ for all $i,j\in [n]$. 
A preference matrices $M$ is $\alpha$-rationalizable if there exists a set $\mathcal{V}$ of $\alpha$-rational voters such that, for any pair of candidates $i$ and $j$,\\
(i) at least a $p_{ij}$ fraction of the voters weakly prefer $i$ over $j$, \\
(ii) at least a $p_{ji}$ fraction of the voters weakly prefer $j$ over $i$. \\
One way to see that this definition accords with $M$ being compatible with the set of voters is via sampling. Suppose we take a large sample of the voters and ask them if they prefer $i$ or $j$. If the voter prefers $i$ over $j$ (namely, $i \succ j$) or vice versa then we insist the voter must declare truthfully. If the voter is indifferent between $i$ and $j$ then we allow the voter to choose either of them.  
Then, in the limit, it is feasible that exactly a $p_{ij}$ fraction of the voters (from the sample) state a preference for $i$ over $j$ if and only if (i) and (ii) hold. 

Observe that since $p_{ij}+p_{ji}=1$, a set of $\alpha$-rational voters
is {\bf not} compatible with $M$ if more than a $p_{ij}=1-p_{ji}$
fraction of the voters strongly prefer $i$ over $j$, or if
more than a $p_{ji}=1-p_{ij}$ fraction of the voters strongly prefer $j$ over $i$.
Consequently, we can encode (i) and (ii) as the following \textit{rationality constraints} for all $i \neq j \in [n]$:
\begin{equation}
\frac{\#\{v \in \V : v \text{ strongly prefers } i \text{ over } j\}}{|\mathcal{V}|} 
\ \leq\  p_{ij}\  \leq\  
\frac{\#\{v \in \V : v \text{ weakly prefers } i \text{ over } j\}}{|\mathcal{V}|} \tag{$\ast$}
    \label{Requirements}
\end{equation}
If $\mathcal{V}$ satisfies (\ref{Requirements}) for a preference matrix $M$ then we say $\mathcal{V}$ is {\em consistent} or {\em compatible} with $M$.
Note that we impose no restriction on the cardinality of $\mathcal{V}$,
we simply desire a voting set of any cardinality that $\alpha$-rationalizes $M$.  

We remark that if every voter only has strong preferences then
(i) and (ii) are equivalent to exactly a $p_{ij}$ fraction of the voters (strongly) preferring $i$ over $j$ (in accordance with {\sc The Matrix Rationalizability Problem}).
Of course, such voters have total order preferences and thus are $1$-rational and, in this case, $M$ would be $1$-rationalizable.
Naturally, in our general setting, we then desire the minimum $\alpha$ such that $M$ is $\alpha$-rationalizable; we call this
minimum the {\em rationality number} of $M$ and denote it by $\alpha(M)$. This induces the following decision problem:\\[0.1cm]
\noindent{\sc The Rationality Number Problem:} Given a preference matrix $M$ and a positive integer $k$, is the rationality number $\alpha(M)$ at most $k$?

\subsection{Examples}
We now present three examples to illustrate these concepts.

\noindent{\tt Example I:}
Consider the integral preference matrix $M$ shown in~Figure \ref{ex_integral}. Observe that there is a simple way to represent
an integral preference matrix by a {\em voting graph}, $D_M = (V, A)$,
whose vertices are the candidates and there is an arc from $i$ to $j$ if and only if $p_{ij}=1$. Thus, as illustrated, the voting graph for $M$ is simply the directed $3$-cycle. More generally an integral preference matrix corresponds to a
{\em tournament}, an orientation of the complete graph on $n$ vertices.

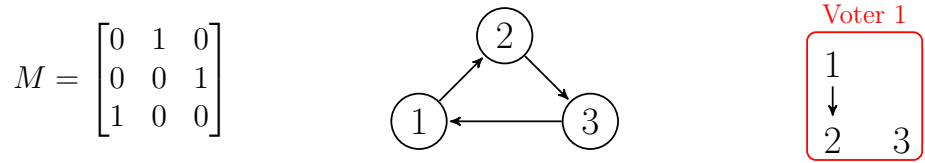
\begin{figure}[h]
\begin{minipage}{0.4\textwidth}
$$M = \begin{bmatrix}
    0 & 1 & 0 \\
    0 & 0 & 1 \\
    1 & 0 & 0
\end{bmatrix}$$
\end{minipage}
\begin{minipage}{0.3\textwidth}
\scalebox{0.85}{
    \begin{tikzpicture}[->,>=stealth',shorten >=1pt,auto,node distance=1.5cm,thick,
                    blank node/.style={circle,draw=black,font=\sffamily\Large\bfseries}]

      \node[blank node] (1) {$1$};
      \node[blank node] (2) [above right=1cm of 1] {$2$};
      \node[blank node] (3) [below right=1cm of 2] {$3$};
      \path
    (1) edge (2)
    (2) edge (3)
    (3) edge (1);
    \end{tikzpicture}}  
\end{minipage} 
\quad
\begin{minipage}{0.25\textwidth}
\scalebox{0.85}{
 \begin{tikzpicture}[->,>=stealth',shorten >=1pt,auto,node distance=2.5cm,thick,
                    blank node/.style={font=\sffamily\Large\bfseries}]

      \node[blank node] (1) {$1$};
      \node[blank node] (2) [below=0.5cm of 1] {$2$};
      \node[blank node] (3) [right=0.5cm of 2] {$3$};
      \path
    (1) edge (2);
     \node[red] () at (0.5,0.5) [above] {Voter 1};
        \draw[rounded corners,red] (-0.4, 0.5) rectangle (1.4, -1.5) {};
    \end{tikzpicture}}
\end{minipage} 
\caption{An integral preference matrix (and its voting graph $D_M$) that is $2$-rationalizable using a single voter.}
\label{ex_integral}
\end{figure}
The matrix $M$ is not rationalizable ($1$-rationalizable). This is because any voter with a total order preference list would strongly
prefer $j$ over $i$, for some arc $(i,j)\in G$, which is incompatible with the fact that $p_{ij}=1$. 
However, $M$ is $2$-rationalizable. Indeed it is compatible with
a single voter with partial order of width $2$ as shown by the red voter in~Figure \ref{ex_integral}.
Specifically, the voter prefers candidate $1$ over $2$ but is indifferent between candidates $1$ and $3$ and 
between candidates $2$ and $3$.

Let's verify that this voter does $2$-rationalize $M$.
We need to prove that the \textit{rationality constraints} are satisfied for every pair of voters. 
Observe $1\le p_{12}=1 \le 1$. Here the lower bound holds 
because the fraction of voters that strongly prefer $1$ over $2$ is 
one (as there is a single voter!).
Furthermore $0\le p_{13}=0 \le 1$ because the voter is indifferent between $1$ and $3$. Thus the fraction of voters that strongly prefer $1$ over $3$ is zero and the fraction that weakly 
prefer $1$ over $3$ is one. Similarly $0\le p_{23}=1 \le 1$.
We remark that if conditions (\ref{Requirements}) hold for $p_{ij}$ then they 
hold for $p_{ji}$. Thus the rationality constraints (\ref{Requirements}) are satisfied for every pair of candidates and $M$ is $2$-rationalizable.

\

\noindent{\tt Example II:}
Consider the half-integral preference matrix $M$ shown in Figure~\ref{ex_half-integral}. Again, we represent a half-integral preference matrix by a {\em voting graph}, $D_M =(V, A)$, where there is an arc from $i$ to $j$ if and only if $p_{ij}=1$. 
Thus if $p_{ij}=\frac12$ there is no arc (in either direction) between
$i$ and $j$. In~Figure~\ref{ex_half-integral}, the voting graph for $M$ is illustrated with a dashed line for the absence of arcs.
We similarly define an undirected {\em unanimity graph} $G_M = (V,E)$, which has an edge between $i$ and $j$ whenever $p_{ij} = 1$ or $p_{ji} = 1$. Thus it contains an edge for each pair of candidates for which the voters cannot strongly disagree. Note that it corresponds to the undirected version of the voting graph $D_M$.

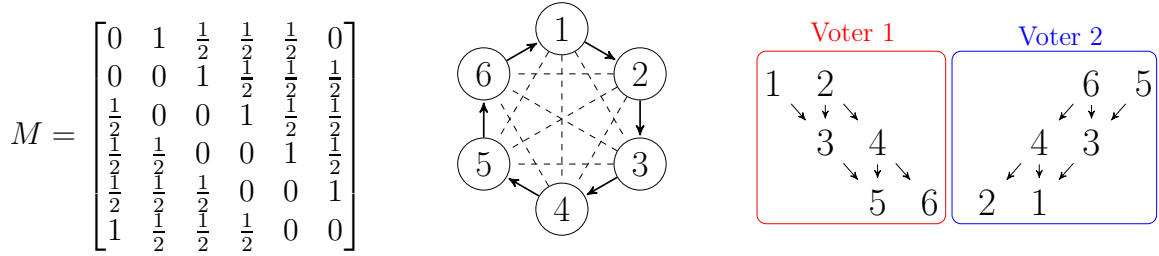
\begin{figure}[h]
\begin{minipage}{0.35\textwidth}
$$M = \begin{bmatrix}
    0 & 1 & \frac{1}{2}&\frac{1}{2} &\frac{1}{2} &0 \\
    0 & 0 & 1 &\frac{1}{2} &\frac{1}{2} &\frac{1}{2}\\
    \frac{1}{2} & 0 & 0 &1 &\frac{1}{2} &\frac{1}{2}\\
    \frac{1}{2}&\frac{1}{2} & 0 &0 &1 &\frac{1}{2}\\
    \frac{1}{2}&\frac{1}{2} &\frac{1}{2} & 0 &0 &1\\
    1&\frac{1}{2} &\frac{1}{2} &\frac{1}{2} & 0 &0
\end{bmatrix}$$
\end{minipage}
\hspace{0.3cm}
\begin{minipage}{0.23\textwidth}
\scalebox{0.8}{
    \begin{tikzpicture}[->,>=stealth',shorten >=1pt,auto,node distance=1.4cm,thin,
                    blank node/.style={circle,draw=black,font=\sffamily\Large\bfseries}]

      \node[blank node] (1) at (90:1.5) {$1$};
      \node[blank node] (2) at (30:1.5) {$2$};
      \node[blank node] (3) at (-30:1.5) {$3$};
      \node[blank node] (4) at (-90:1.5) {$4$};
      \node[blank node] (5) at (-150:1.5) {$5$};
      \node[blank node] (6) at (150:1.5) {$6$};
      \path
    (1) edge [thick] (2)
    (2) edge [thick](3)
    (3) edge [thick](4)
    (4) edge [thick](5)
    (5) edge [thick](6)
    (6) edge [thick](1)
    (1) edge [-,dashed] (3)
    (1) edge [-,dashed] (4)
    (1) edge [-,dashed] (5)
    (2) edge [-,dashed] (4)
    (2) edge [-,dashed] (5)
    (2) edge [-,dashed] (6)
    (3) edge [-,dashed] (6)
    (3) edge [-,dashed] (5)
    (4) edge [-,dashed] (6);
    \end{tikzpicture}}
\end{minipage}
\begin{minipage}{0.25\textwidth}
\scalebox{0.85}{
   \begin{tikzpicture}[->,>=stealth',shorten >=1pt,auto,node distance=2cm,thin,
                        blank node/.style={font=\sffamily\Large\bfseries}]

        \node[blank node] (1) at (0,0) {$1$};
        \node[blank node] (2) [right=.25cm of 1]  {$2$};
        \node[blank node] (3) [below=.25cm of 2] {$3$};
        \node[blank node] (4) [right=.25cm of 3] {$4$};
        \node[blank node] (5) [below=.25cm of 4] {$5$};
        \node[blank node] (6) [right=.25cm of 5] {$6$};
        \node[blank node] (7) at (5.8,0) {$5$};
        \node[blank node] (8) [left=.25cm of 7]  {$6$};
        \node[blank node] (9) [below=.25cm of 8] {$3$};
        \node[blank node] (10) [left=.25cm of 9] {$4$};
        \node[blank node] (11) [below=.25cm of 10] {$1$};
        \node[blank node] (12) [left=.25cm of 11] {$2$};

        \node[red] () at (1.25,0.8) {Voter 1};
        \draw[rounded corners, red] (-0.25, 0.5) rectangle (2.7, -2.2) {};
        \node[blue] () at (4.5,0.75) {Voter 2};
        \draw[rounded corners,blue] (2.8, 0.5) rectangle (6, -2.2) {};
        
        \path
            (1) edge (3)
            (2) edge (3)
            (2) edge (4)
            (3) edge (5)
            (4) edge (6)
            (4) edge (5)
            (7) edge (9)
            (8) edge (9)
            (8) edge (10)
            (9) edge (11)
            (10) edge (12)
            (10) edge (11);
    \end{tikzpicture}}
\end{minipage} 
\caption{A half-integral preference matrix (and its voting graph $D_M$) that is $2$-rationalizable using two voters.}
\label{ex_half-integral}
\end{figure}

The matrix $M$ is not rationalizable ($1$-rationalizable). Again, this is because the voting graph contains a directed cycle $C$ on the six candidates. Thus, any voter with a total order preference list would strongly prefer $j$ over $i$, for at least one arc $(i,j)\in C$, which is incompatible with the fact that $p_{ij}=1$. 

However, $M$ is $2$-rationalizable. But,
unlike in Example I, this requires at least two voters.
To see that one voter is insufficient, consider the three candidates $\{1,3,5\}$. For any $2$-rational voter $v$, at least two of these must be comparable in its partial order; otherwise they form an antichain
of cardinality $3$. Without loss of generality, let $1\succ_v 3$.
So if $v$ is the only voter then the proportion of voters that prefer
$1$ over $3$ is one. This contradicts condition (\ref{Requirements})
because $p_{13}=\frac{1}{2}$.

On the other hand, $M$ is compatible with two voters that are each $2$-rational as illustrated by the red and blue voters in Figure~\ref{ex_half-integral}.
For example, both the red and blue voters are indifferent between candidates $1$ and $2$. 
Thus the fraction of voters that strongly prefer $1$ over $2$ is zero and the fraction that weakly 
prefer $1$ over $2$ is one. Thus condition (\ref{Requirements})
holds for this pair as $0\le p_{12}=1 \le 1$. Next, the red voter prefers $1$ over $3$ but the blue voter prefers $3$ over $1$.
Thus the fraction of voters that strongly prefer $1$ over $3$ is half and the fraction that weakly 
prefer $1$ over $3$ is also half. Thus condition (\ref{Requirements})
holds for this pair as $\frac12\le p_{13}=\frac12 \le \frac12$.
Further, the red voter is indifferent between $1$ and $6$ but the blue voter prefers $6$ over $1$.
Thus the fraction of voters that strongly prefer $1$ over $6$ is zero and the fraction that weakly 
prefer $1$ over $6$ is half. Thus condition (\ref{Requirements})
holds for this pair as $0\le p_{16}= 0 \le \frac12$.
The reader may verify that the conditions (\ref{Requirements}) also hold for every other pair of candidates.
Hence $M$ is $2$-rationalizable, as claimed.

\

\noindent{\tt Example III:}
Finally consider the generic preference matrix $M$ with three candidates shown in Figure~\ref{ex_general}.
$M$ is $3$-rationalizable using a single voter whose partial order is
an antichain on all the candidates.
Observe that since the voter has no strict preference, the fraction
of voters that strongly prefer $i$ over $j$ is zero, for any pair of candidates. Similarly, the fraction of voters that weakly prefer $i$ over $j$ is one. Thus the rationality constraints (\ref{Requirements}) are simply
$0\le p_{ij} \le 1$ which are trivially satisfied. 
Thus $M$ is $3$-rationalizable.

\begin{figure}[h]
\begin{minipage}{0.55\textwidth}
$$M = \begin{bmatrix}
    0 & p_{12} & p_{13} \\
    1-p_{12} & 0 & p_{23} \\
   1-p_{13} & 1-p_{23} & 0
\end{bmatrix}$$
\end{minipage}
\hspace{0.5cm}
\begin{minipage}{0.4\textwidth}
\scalebox{0.85}{
 \begin{tikzpicture}[->,>=stealth',shorten >=1pt,auto,node distance=2.5cm,thick,
                    blank node/.style={font=\sffamily\Large\bfseries}]

      \node[blank node] (1) {$1$};
      \node[blank node] (2) [right=1cm of 1] {$2$};
      \node[blank node] (3) [right=1cm of 2] {$3$};

     \node[blue] () at (1.5,0.5) [above] {Voter 1};
        \draw[rounded corners,blue] (-0.5, 0.5) rectangle (3.75, -.5) {};
    \end{tikzpicture}}
\end{minipage} 
\caption{A generic preference matrix that is $3$-rationalizable using a single voter.}
\label{ex_general}
\end{figure}
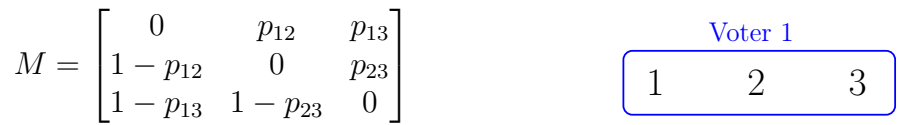
Of course, this example, trivially generalizes to any number $n$ of candidates. A single voter whose partial order is an antichain of size $n$ will $n$-rationalize any preference matrix with $n$ candidates.
\begin{observation}\label{obs:exist}
A preference matrix with $n$ candidates is $n$-rationalizable.
\end{observation}
Observation~\ref{obs:exist} confirms the existence of a set of voters that $\alpha$-rationalize
a preference matrix $M$. Of course, our interest is whether or not the matrix is $\alpha$-rationalizable for some $\alpha$ much smaller than $n$.

\subsection{Our Results}
In Section~\ref{Sec:prelim} we present structural results
concerning $\alpha$-rationalizable preference matrices.
Then in Section~\ref{sec:half}, we focus on the important
class $\mathcal{M}^{\frac12}$ of half-integral preference matrices. Our first main result is that, for this class, the rationality number $\alpha(M)$ is bounded by the
chromatic number of the (undirected) unanimity graph $G_M$.
Specifically, we prove:

\begin{theorem}\label{thm:half-integral}
Let $\mathcal{M}^{\frac12}$ be the class of half-integral preference matrices. Then
$$\frac15 \c(G_M) \ \le_{\exists}\  \alpha(M) \ \le_{\forall}\ \c(G_M)$$
\end{theorem}
In order to concisely formulate our results, we use the notation $\leq_\exists$ and $\leq_\forall$.
Here $\leq_\exists$ means that for every $k \in \N$
{\em there exists} a preference matrix $M$ in that class (here, $\mathcal{M}^{\frac12}$) with chromatic number $k$ such that the inequality is satisfied, and
$\leq_\forall$ means that the inequality is satisfied {\em for every} preference matrix $M$ in that class.

Next, in Section~\ref{sec:integral}, we strengthen our results
for the class $\mathcal{M}^{0/1}$ of integral preference matrices. For this class, we prove that the rationality 
number $\alpha(M)$ is equal to the dichromatic number of the (directed) voting graph $D_M$. We use this to give even more precise bounds on the rationality number for the class of integral matrices. Specifically our second main result is:
\begin{theorem}\label{thm:integral}
Let $\mathcal{M}^{0/1}$ be the class of integral preference matrices. Then
$$\frac{\c(G_M)}{2\log n+1} \ \le_{\exists}\ \alpha(M) \ \le_{\forall}\  \frac{3\c(G_M)}{\log n}$$
\end{theorem}
Note that for an integral preference matrix $M$ its unanimity graph $G_M$ is the complete graph and thus has chromatic number exactly $\c(G_M)=n$. 

We conclude, in Section~\ref{sec:complexity}, by showing that the rationality number problem is NP-complete, even for the class of integral matrices. 

\subsection{Literature Review}\label{sec:literature}

The matrix rationalizability problem asks when a binary probability system (that is, a preference matrix) corresponds to a distribution over total orders. This problem dates back to the 1950s and the works of Guilbaud~\cite{Gui53} and Marschak~\cite{Mar60}.
Early works proved that the triangle inequality\footnote{The triangle inequality states that $p_{ij}\le p_{ik} + p_{kj}$, for any three candidates $i,j$ and~$k$.} is a necessary and sufficient condition for a preference matrix to be rationalizable in the case of five or fewer candidates~\cite{Dri80,Fishburn87}.
But this characterization fails for six or more candidates
\cite{Meg77,Dri80,Fishburn87,Menezes},
leading to a search for other necessary conditions~\cite{GJR83,Rein85}.
Obtaining a concise characterization for matrix rationalizability remains an outstanding open problem.

{\em Systems of choice probabilities}
have been widely studied in mathematical psychology.
Specifically, in stochastic choice behaviour,
decision-makers must select an item $i$ when presented with 
a subset $S$ of the items. Binary probability systems are the special case where the subsets considered have size two.
Falmagne~\cite{Fal78} showed a complete system of choice probabilities to be induced by rankings whenever it satisfies the Block–Marschak conditions~\cite{BM60} and normalization equalities, the proof of which was simplified by Fiorini~\cite{Fio04}.

Closely related to the matrix rationalizability problem is the {\em linear ordering problem}, which asks for the total order that best approximates a given binary probability system.
For standard measures this problem is NP-hard~\cite{Garey1990}, but understanding the geometry of the polytope of rationalizable choice matrices aids in the development of heuristic and approximation algorithms~\cite{Fishburn92,MR11}. The linear ordering problem has many applications in data science~\cite{CCP21} and economics~\cite{GJR83}.

Also closely related to the matrix rationalizability problem
is the classic combinatorial {\em majority digraph problem}. Here a directed graph models the pairwise majority relation of a voter profile; that is, the arc $ij$ is included in the digraph whenever the majority of voters prefer candidate $i$ over candidate $j$. 
McGarvey~\cite{McGarvey53} showed that every asymmetric digraph corresponds to the majority relation of some voter profile. Stearns~\cite{Stearns59} and Erdős and Moser \cite{EM64} improved the bound on the number of voters required in such a voter profile to $\Theta(n/ \log n)$. Recently, Bachmeier et al.~\cite{Bachmeier19} characterized the digraphs induced by a constant number of voters and studied the computational complexity of determining the minimum number of voters required to induce a given digraph.

Various relaxations of rationalizability have been studied. One popular relaxation is {\em regularity}~\cite{DP07,Suck16}
which asks if a system of binary probabilities can be extended to a system of choice probabilities under which the probability of selecting an item from a set does not increase when that set expands.

\section{Preliminaries}\label{Sec:prelim}
We begin with a monotonicity property.
\begin{lemma}\label{lem:mono}
Let $M$ be a preference matrix consistent with $m$ voters with preferences $\{\succ_1,\dots,\succ_u,\dots,\succ_m\}$.
Then $M$ is consistent with $\{\succ_1,\dots,\succ'_u,\dots \succ_m\}$, where $\succ'_u$ is identical to $\succ_u$
except that voter $u$ prefers $x$ over $y$ in $\succ_u$ but 
is indifferent between $x$ and $y$ in $\succ'_v$.
\end{lemma}
\begin{proof}
So, for any pair of candidates $i$ and $j$, $\{\succ_1,\dots,\succ_u,\dots,\succ_m\}$ satisfies the conditions (\ref{Requirements}). Namely
\begin{equation*}
\frac{\#\{v: v \text{ strongly prefers } i \text{ over } j\}}{m} 
\ \leq\  p_{ij}\  \leq\  
\frac{\#\{v: v \text{ weakly prefers } i \text{ over } j\}}{m} 
\end{equation*}
These conditions trivially still hold with respect to 
$\{\succ_1,\dots,\succ'_u,\dots \succ_m\}$ for any pair except $\{x,y\}$ and $\{y, x\}$. Let's verify that (\ref{Requirements}) still holds for these two cases as well. 
As $u$ prefers $x$ over $y$ in $\succ_u$ but 
is indifferent between $x$ and $y$ in $\succ'_v$, we have that
the number of voters that strongly prefer $x$ over $y$ has fallen
by one (namely, $u$) while the number of voters that weakly prefer $x$ over $y$ is the same.
Thus the lower bound has fallen whilst the upper bound is identical.
Hence (\ref{Requirements}) holds for $\{x,y\}$.

On the other hand, consider $\{y, x\}$. 
Now the number of voters that strongly prefer $y$ over $x$ is the same
while the number of voters that weakly prefer $y$ over $x$ has increased by one (namely, $u$).
Thus the lower bound is the same whilst the upper bound has increased.
Hence (\ref{Requirements}) holds for $\{x,y\}$ and 
$M$ is consistent with $\{\succ_1,\dots,\succ'_u,\dots \succ_m\}$.
\end{proof}

Take a finite poset $\mathcal{P}=(S,\succ)$ on a set $S$ of elements
with partial order $\succ$.
A {\em chain decomposition} of $\mathcal{P}$ is a partition of the elements of the poset into disjoint chains. The cardinalility of a chain decomposition is the number of chains in the decomposition.
A famous result of Dilworth~\cite{Dil50} states that the width of $\mathcal{P}$
is the minimum cardinality of a chain decomposition.
\begin{theorem}\cite{Dil50}\label{thm:Dilworth}
Let $\mathcal{P}=(S,\succ)$ be a finite poset. The maximum cardinality of an antichain of $\mathcal{P}$ equals the minimum cardinality of a chain decomposition of~$\mathcal{P}$. \qed
\end{theorem}
Theorem~\ref{thm:Dilworth} allows us to restrict our attention to voter preferences composed of disjoint chains.
\begin{theorem}\label{thm:chains}
Let $M$ be $\alpha$-rationalizable by $\{\succ_1,\dots,\succ_u,\dots,\succ_m\}$.
Then $M$ is $\alpha$-rationalizable by $\{\succ_1,\dots,\succ'_u,\dots \succ_m\}$, where $\succ'_u$ consists of at most $\alpha$ disjoint chains.
\end{theorem}
\begin{proof}
Let $\succ'_u$ correspond to a minimum cardinality chain decomposition of
$\mathcal{P}=([n],\succ_u)$. Thus, by Theorem~\ref{thm:Dilworth}, 
$\succ'_u$ consists of at most $\alpha$ disjoint chains.
Thus the width of $\succ_u$ and $\succ'_u$ are the same, and voter $u$ is still $\alpha$-rational.

Next we know $M$ is consistent with $\{\succ_1,\dots,\succ_u,\dots \succ_m\}$.
We can now use the monotonicity property.
Repeatedly applying Lemma~\ref{lem:mono}, we conclude that 
$M$ is consistent with $\{\succ_1,\dots,\succ'_u,\dots \succ_m\}$,
as desired.
\end{proof}

Of course, repeated application of Theorem~\ref{thm:chains} implies we can assume that every voter has a partial order that is a collection of disjoint 
chains.
\begin{corollary}\label{cor:chains}
If $M$ is $\alpha$-rationalizable then it is consistent with a set of $\alpha$-rational voters whose partial orders are 
a collection of (at most) $\alpha$ chains.~\qed
\end{corollary}

For an application, consider again {\tt Example II}. Recall the preference matrix $M$ is $2$-rationalizable using two voters. By Corollary~\ref{cor:chains}, it must then be consistent
with two $2$-rational voters whose partial orders each consist of two disjoint chains. To do this we simply find a minimum chain decomposition of partial orders for the red and blue voters in {\tt Example II}. This gives us the two voters illustrated in Figure~\ref{ex_chains}.

\begin{figure}[h]
\begin{minipage}{0.4\textwidth}
$$M = \begin{bmatrix}
    0 & 1 & \frac{1}{2}&\frac{1}{2} &\frac{1}{2} &0 \\
    0 & 0 & 1 &\frac{1}{2} &\frac{1}{2} &\frac{1}{2}\\
    \frac{1}{2} & 0 & 0 &1 &\frac{1}{2} &\frac{1}{2}\\
    \frac{1}{2}&\frac{1}{2} & 0 &0 &1 &\frac{1}{2}\\
    \frac{1}{2}&\frac{1}{2} &\frac{1}{2} & 0 &0 &1\\
    1&\frac{1}{2} &\frac{1}{2} &\frac{1}{2} & 0 &0
\end{bmatrix}$$
\end{minipage}
\hspace{0.2cm}
\begin{minipage}{0.25\textwidth}
\scalebox{0.8}{
    \begin{tikzpicture}[->,>=stealth',shorten >=1pt,auto,node distance=1.4cm,thin,
                    blank node/.style={circle,draw=black,font=\sffamily\Large\bfseries}]

      \node[blank node] (1) at (90:1.5) {$1$};
      \node[blank node] (2) at (30:1.5) {$2$};
      \node[blank node] (3) at (-30:1.5) {$3$};
      \node[blank node] (4) at (-90:1.5) {$4$};
      \node[blank node] (5) at (-150:1.5) {$5$};
      \node[blank node] (6) at (150:1.5) {$6$};
      \path
    (1) edge (2)
    (2) edge (3)
    (3) edge (4)
    (4) edge (5)
    (5) edge (6)
    (6) edge (1)
    (1) edge [-,dashed] (3)
    (1) edge [-,dashed] (4)
    (1) edge [-,dashed] (5)
    (2) edge [-,dashed] (4)
    (2) edge [-,dashed] (5)
    (2) edge [-,dashed] (6)
    (3) edge [-,dashed] (6)
    (3) edge [-,dashed] (5)
    (4) edge [-,dashed] (6);
    \end{tikzpicture}}
\end{minipage} 
\begin{minipage}{0.25\textwidth}
\scalebox{0.85}{
   \begin{tikzpicture}[->,>=stealth',shorten >=1pt,auto,node distance=2.5cm,thick,
                        blank node/.style={font=\sffamily\Large\bfseries}]

        \node[blank node] (1) at (0,0) {$1$};
        \node[blank node] (2) [right=.5cm of 1]  {$2$};
        \node[blank node] (3) [below=.5cm of 1] {$3$};
        \node[blank node] (4) [right=.5cm of 3] {$4$};
        \node[blank node] (5) [below=.5cm of 3] {$5$};
        \node[blank node] (6) [right=.5cm of 5] {$6$};
        \node[blank node] (7) at (2,0) {$5$};
        \node[blank node] (8) [right=.5cm of 7]  {$6$};
        \node[blank node] (9) [below=.5cm of 7] {$3$};
        \node[blank node] (10) [right=.5cm of 9] {$4$};
        \node[blank node] (11) [below=.5cm of 9] {$1$};
        \node[blank node] (12) [right=.5cm of 11] {$2$};

        \node[red] () at (0.6,0.75) {Voter 1};
        \draw[rounded corners, red] (-0.3, 0.5) rectangle (1.5, -2.75) {};
        \node[blue] () at (2.6,0.75) {Voter 2};
        \draw[rounded corners,blue] (1.7, 0.5) rectangle (3.5, -2.75) {};
        
        \path
            (1) edge (3)
            (2) edge (4)
            (3) edge (5)
            (4) edge (6)
            (7) edge (9)
            (8) edge (10)
            (9) edge (11)
            (10) edge (12);
    \end{tikzpicture}}
\end{minipage} 
\caption{A $2$-rationalizable matrix consistent with two voters whose partial orders are disjoint chains.}
\label{ex_chains}
\end{figure}
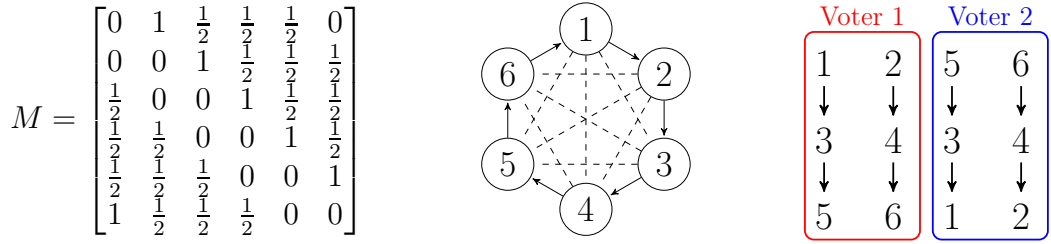

\section{Half-Integral Preference Matrices}\label{sec:half}

We now restrict attention to the class $\mathcal{M}^{\frac12}$ of half-integral matrices. This class is important in the field as it is the simplest class of preference matrices for which no characterization of rationality is known.\footnote{There is a simple characterization for the class $\mathcal{M}^{0/1}$ of integral matrices. An integral preference matrix $M$
is rationalizable ($1$-rationalizable) if and only if its
voting graph $D_M$ is acyclic.}
Recall, associated with a half-integral preference matrix is
a directed {\em voting graph} $D_M$ 
and an undirected {\em unanimity graph} $G_M$.
In this section, we will see that understanding the 
chromatic number of the unanimity graph is critical in
understanding the $\alpha$-rationalizability of the class of
half-integral preference matrices.

\subsection{Rationality and the Unanimity Graph}
We begin with a simple reduction that allows us to restrict the computation of the rationality number
of $M$ to its computation in a collection of submatrices based upon the structure of the unanimity graph $G_M$.

\begin{theorem}\label{thm:components}
Given a half-integral preference matrix $M$ with connected components  $G_1, G_2, \dots, G_t$ in the unanimity graph $G_M$. For each, $1\le \ell \le t$, let $M_{\ell}$ be the sub-matrix of $M$ induced by 
$G_{\ell}$. Then
\begin{equation*}\label{eq:components}
\alpha(M) \ =\  \max_{\ell} \, \alpha(M_\ell)
\end{equation*}
\end{theorem}
\begin{proof}
Given a set of $\alpha$-rational voters that satisfy the conditions (\ref{Requirements}) for every pair of candidates in $M$. Then, for each $1\le \ell \le t$, the same set of voters trivially satisfy  (\ref{Requirements}) for each pair of candidates in $V_\ell$, the set of candidates restricted to the sub-matrix $M_{\ell}$. Thus $\alpha(M) \ge \alpha(M_\ell)$
and, hence, $\alpha(M)\ge \max\limits_{{\ell}} \alpha(M_\ell)$.

So it remains to prove the harder direction, that 
$\alpha(M)\le \max\limits_{\ell} \alpha(M_\ell)$.
Assume, for each $1\le \ell \le t$, that $M_{\ell}$ is consistent with a collection of $m_{\ell}$ voters who are each $\alpha$-rational.
We claim that $M$ is consistent with a collection of $2\cdot\prod_{\ell=1}^t m_{\ell}$ voters who are each $\alpha$-rational.

To prove this, we create two new voters, $L^S$ and $R^S$, for every set $S=\{v_1,v_2, \dots, v_t\}$ of voters, where $v_\ell$ is one of the $n_{\ell}$ voters used to $\alpha$-rationalize $M_{\ell}$.
Note that voter $v_{\ell}$ has a partial order on the candidates in $V_{\ell}$.
Both voters $L^S$ and $R^S$ will copy the partial order $v_\ell$ has on the set $V_\ell$, for all $1\le \ell \le t$.
That is if candidates $i$ and $j$ are both in $V_\ell$
then $L^S$ and $R^S$ comparatively rank $i$ and $j$ exactly how $v_\ell$ does.

But what if candidate $i \in V_{\ell}$ and candidate $j \in V_{\gamma}$, where $\ell \neq \gamma$?
Imagine an ordering of the sets of candidates $\{V_1, V_2,\dots, V_t\}$
from left to right. Then voter $L^S$ will prefers sets from left to right, and voter $R^S$ will prefers sets from right to left.
That is, if $\gamma < \ell$, then $L^S$ prefers $i$ over $j$
and $R^S$ prefers $j$ over $i$.

So in total we have created $2\cdot\prod_{\ell=1}^t m_{\ell}$ new voters. Moreover each of these new voters is $\alpha$-rational.
This is because both $L^S$ and $R^S$ have a strict preference
for any pair of candidates for $i \in V_{\ell}$ and candidate $j \in V_{\gamma}$, where $\ell \neq \gamma$.
Hence, any antichain of cardinality greater than one can only contain
candidates within the same set $V_{\ell}$.
Thus, the maximum size of an antichain in the partial order of $L^S$
(or $R^S$) is equal to the maximum size of an antichain in any of the partial orders for the set of voters $S=\{v_1,\dots, v_t\}$, which by definition is at most $\alpha$. So $L^S$ and $R^S$ are both $\alpha$-rational, for any set $S$.

Finally it remains to prove that the constraints 
(\ref{Requirements}) hold for every pair of candidates using the 
$2\cdot\prod_{\ell=1}^t m_{\ell}$ new voters.
First, take a pair of candidates $i \in V_{\ell}$ and candidate $j \in V_{\gamma}$, where $\ell \neq \gamma$.
Then since $G_{\ell}$ and $G_{\gamma}$ separate components in the 
unanimity graph $G_M$ it follows that $p_{ij}=\frac12$.
Moreover for any set $S=\{v_1, v_2,\dots, v_t\}$ 
$L^S$ and $R^S$ rank $i$ and $j$ in the opposite way.
Thus exactly half the $2\cdot\prod_{\ell=1}^t m_{\ell}$ voters strongly prefer $i$ over $j$ and half strongly prefer $j$ over $i$.
It follows that (\ref{Requirements}) holds for this pair of candidates.
Second, take a pair of candidates $i, j \in V_{\ell}$.
Recall there are $m_{\ell}$ $\alpha$-rational voters consistent with
the submatrix $M_{\ell}$. Therefore $f_1 \le p_{ij}\le f_2$,
where $f_1$ is the fraction of these $m_{\ell}$ voters that strongly prefer $i$ over $j$, and $f_2$ is the fraction of these voters that weakly prefer $i$ over $j$. But each voter in $M_{\ell}$ is selected to be $v_{\ell}$ in $S=\{v_1, v_2,\dots, v_t\}$
with exactly the same probability, namely $\frac{1}{m_{\ell}}$.
It immediately follows that among the $2\cdot\prod_{\ell=1}^t m_{\ell}$ new voters exactly an $f_1$ fraction of them strongly prefer $i$ over $j$, and exactly an $f_2$ fraction of them weakly prefer $i$ over $j$. Thus (\ref{Requirements}) holds, and $M$ is indeed
$\alpha$-rationalizable. So $\alpha(M)\ge \max\limits_{\ell} \alpha(M_\ell)$.
\end{proof}

\begin{figure}[h]
\begin{minipage}{0.35\textwidth}
$$M = \begin{bmatrix}
    0 & 1 & 0 & \frac{1}{2} & \frac{1}{2} \\
    0 & 0 & 1 & \frac{1}{2} & \frac{1}{2} \\
    1 & 0 & 0 & \frac{1}{2} & \frac{1}{2} \\
    \frac{1}{2} & \frac{1}{2} & \frac{1}{2} & 0 & 1 \\
    \frac{1}{2} & \frac{1}{2} & \frac{1}{2} & 0 & 0
\end{bmatrix}$$
\end{minipage}
\hspace{0.4cm}
\begin{minipage}{0.3\textwidth}
\scalebox{0.85}{
    \begin{tikzpicture}[->,>=stealth',shorten >=1pt,auto,node distance=2cm,thick,
                    blue node/.style={circle,draw=blue,fill=blue!20,font=\sffamily\Large\bfseries}]

      \node[blue node] (1) {$1$};
      \node[blue node] (2) [above right=0.8cm of 1] {$2$};
      \node[blue node] (3) [below right=0.8cm of 2] {$3$};
      \node[blue node] (4) [below=0.7cm of 1] {$4$};
      \node[blue node] (5) [below=0.7cm of 3] {$5$};
      \path
    (1) edge (2)
    (2) edge (3)
    (3) edge (1)
    (4) edge (5)
    (1) edge [-,dashed] (4)
    (1) edge [-,dashed] (5)
    (2) edge [-,dashed] (4)
    (2) edge [-,dashed] (5)
    (3) edge [-,dashed] (4)
    (3) edge [-,dashed] (5);
    \end{tikzpicture}}
\end{minipage} 
\begin{minipage}{0.3\textwidth}
\scalebox{0.85}{
    \begin{tikzpicture}[>=stealth',shorten >=1pt,auto,node distance=2cm,thick,
                    blue node/.style={circle,draw=blue,fill=blue!20,font=\sffamily\Large\bfseries}]

      \node[blue node] (1) {$1$};
      \node[blue node] (2) [above right=0.8cm of 1] {$2$};
      \node[blue node] (3) [below right=0.8cm of 2] {$3$};
      \node[blue node] (4) [below=0.7cm of 1] {$4$};
      \node[blue node] (5) [below=0.7cm of 3] {$5$};
    \draw (1)-- (2) -- (3) -- (1);
    \draw (4) -- (5);
    \end{tikzpicture}}
\end{minipage} 
    \caption{A half-integral preference matrix with its voting graph and unamimity graph.}
    \label{ex_reduction_cc1}
\end{figure}
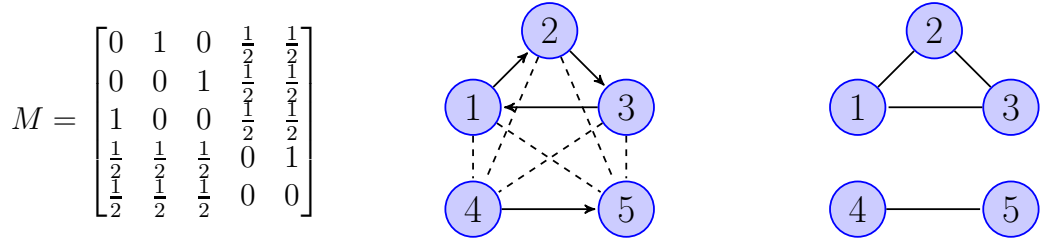
\noindent{\tt Example IV.}
Consider the half-integral preference matrix $M$ shown in Figure~\ref{ex_reduction_cc1}, along with its voting graph $D_M$ and unanimity graph $G_M$.
Observe that the unanimity graph has exactly two components
on the candidate sets $V_1=\{1,2,3\}$ and $V_2=\{4,5\}$.
We can prove that $M$ is $2$-rationalizable by applying the
method of Theorem~\ref{thm:components}. The submatrix $M_1$ induced by $V_1$ simply corresponds to the $3$-cycle we saw in {\tt Example I}. Thus $M_1$ is $2$-rationalizable with $m_1=1$ voter with preference $1$ over $2$ and
an indfference between candidate $3$ and the other two candidates.
The submatrix $M_2$ induced by $V_2$ corresponds to an arc $(4,5)$.
This is trivially $1$-rationalizable (and, hence, $2$-rationalizable) with $m_2=1$ voter with preference $4$ over $5$.
So we only need $2\cdot m_1\cdot m_2 = 2$ new voters to $2$-rationalize $M$. Let these voters be $L$ and $R$. Following the proof of Theorem~\ref{thm:components}, let $L$ and $R$ copy the preferences within each of $V_1=\{1,2,3\}$ and $V_2=\{4,5\}$. 
Now $L$ prefers any vertex in $V_1$ over any vertex in $V_2$,
but $R$ has the opposite preference. Thus we obtain the two voters shown in Figure~\ref{ex_LR}.

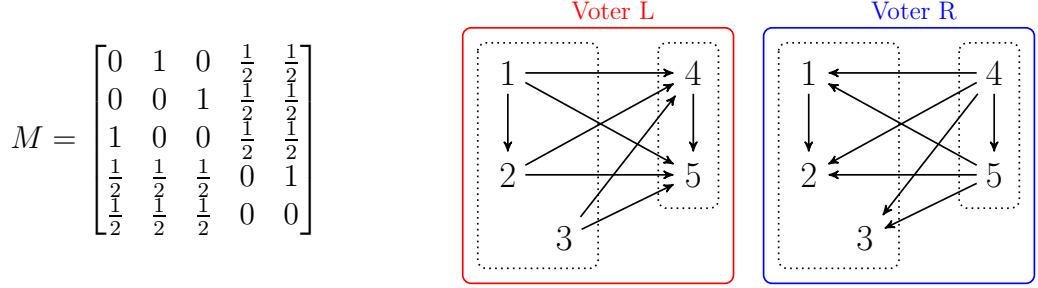
\begin{figure}[h]
\begin{minipage}{0.45\textwidth}
$$M = \begin{bmatrix}
    0 & 1 & 0 & \frac{1}{2} & \frac{1}{2} \\
    0 & 0 & 1 & \frac{1}{2} & \frac{1}{2} \\
    1 & 0 & 0 & \frac{1}{2} & \frac{1}{2} \\
    \frac{1}{2} & \frac{1}{2} & \frac{1}{2} & 0 & 1 \\
    \frac{1}{2} & \frac{1}{2} & \frac{1}{2} & 0 & 0
\end{bmatrix}$$
\end{minipage}
\begin{minipage}{0.65\textwidth}
\scalebox{0.8}{
 \begin{tikzpicture}[->,>=stealth',shorten >=1pt,auto,node distance=2.5cm,thick,
                    blank node/.style={font=\sffamily\Large\bfseries}]

      \node[blank node] (1) {$1$};
      \node[blank node] (2) [below=1cm of 1] {$2$};
      \node[blank node] (3) [below right=.5cm of 2] {$3$};
      \node[blank node] (4) [right=2.5cm of 1] {$4$};
      \node[blank node] (5) [below=1cm of 4] {$5$};
      \path
    (1) edge (2) 
    (1) edge (4)
    (1) edge (5)
    (2) edge (4)
    (2) edge (5)
    (3) edge (4)
    (3) edge (5)
      (4) edge (5);
     \node[red] () at (1.75, 0.75) [above] {Voter L};
        \draw[rounded corners,red] (-0.75, 0.75) rectangle (3.75, -3.5) {};
           \draw[rounded corners, dotted, black] (-0.5, 0.5) rectangle (1.5, -3.25) {};
        \draw[rounded corners, dotted, black] (2.5, 0.5) rectangle (3.5, -2.25) {};
        
   \node[blank node] (6) at (5,0) {$1$};
      \node[blank node] (7) [below=1cm of 6] {$2$};
      \node[blank node] (8) [below right=.5cm of 7] {$3$};
      \node[blank node] (9) [right=2.5cm of 6] {$4$};
      \node[blank node] (10) [below=1cm of 9] {$5$};
      \path
    (6) edge (7) 
    (9) edge (6)
    (10) edge (6)
    (9) edge (7)
    (10) edge (7)
    (9) edge (8)
    (10) edge (8)
      (9) edge (10);
     \node[blue] () at (6.75, 0.75) [above] {Voter R};
        \draw[rounded corners,blue] (4.25, 0.75) rectangle (8.75, -3.5) {};
           \draw[rounded corners, dotted, black] (4.5, 0.5) rectangle (6.5, -3.25) {};
        \draw[rounded corners, dotted, black] (7.5, 0.5) rectangle (8.5, -2.25) {}; 
    \end{tikzpicture}}
    \end{minipage}
\caption{The matrix $M$ is $2$-rationalizable using two voters.}
\label{ex_LR}
\end{figure}

\subsection{An Upper Bound on the Rationality Number of Half-Integral Matrices}

We can also use the unanimity graph to bound the rationality number of half-integral matrices. 
First we give an upper bound. 
The rationality number $\alpha(M)$ is upper bound by the
chromatic number $\c(G_M)$ of its unanimity graph.

\begin{lemma}\label{lem:upper-half}
If $M$ is a half-integral preference matrix then
$\alpha(M) \le \c(G_M)$.
\end{lemma}
\begin{proof}
Given $M$ let $\c(G_M)=k$ be the chromatic number of its unanimity graph $G_M$.
Take any $k$-colouring of the candidates, namely the vertices in $G_M$. Let $C_\ell$ be the set of candidates receiving colour $\ell$, for $1\le \ell\le k$.
We will show that $M$ is $k$-rationalizable using just two voters. The construction is simple. 
Both voters will have $k$ chains in their partial order and are thus $k$-rational. There is a chain for each colour class $C_\ell$. Voter $1$ places the candidates in $C_\ell$ in an arbitrary total order to generate its $\ell$th chain, for $1\le \ell\le k$.
Voter $2$ does the same thing, except it chooses exactly the opposite total order for the candidates in $C_\ell$ to generate its $\ell$th chain. 

Let's verify that this construction satisfies the rationality constraints
(\ref{Requirements}). Take a pair of candidates $i$ and $j$. There are two cases to consider.
First, assume that both candidates belong to the same colour class, $i, j\in C_{\ell}$.
Because each colour class is an independent set we have that
$(i,j)\notin G_M$ and so $p_{ij}=\frac12$.
But, without loss of generality, voter $1$ prefers $i$ over $j$ and voter $2$ prefers $j$ over $i$ as their chains are reversals of each other. Hence, half of the voters strongly prefer 
$i$ over $j$ and half of the voters weakly prefer 
$i$ over $j$. So $\frac12 \le p_{ij}=\frac12\le \frac12$ and
(\ref{Requirements}) holds for this pair of candidates.

Second, assume the candidates belong to different colour classes, $i\in C_{\ell}$ and $j\in C_{\gamma}$, where
 $\ell\neq \gamma$. But this means $i$ and $j$ are in different chains in the partial orders of both voter $1$ and voter $2$.
 Consequently, the fraction of voters that strongly prefer 
$i$ over $j$ is zero and the fraction of voters that weakly prefer $i$ over $j$ is one. Thus, regardless of the value of $p_{ij}$, we have $0 \le p_{ij}=\frac12\le 1$ and
(\ref{Requirements}) holds for this pair of candidates.
Thus $M$ is $k$-rationalizable and so $\alpha(M) \le \c(G_M)$.
\end{proof}

\noindent{\tt Example V.}
Consider the half-integral preference matrix $M$ shown in Figure~\ref{ex_chromatic_number}.
This has a two-colourable unanimity graph, $\c(G_M)=2$.
It has colour classes green and pink with $V_{green}=\{2,3,5\}$ 
and $V_{pink}=\{1,4\}$. Let voter $1$ order the two corresponding chains by preferring lower numbered candidates,
and let voter prefer higher numbered candidates.
This produces the partial orders illustrated in Figure~\ref{ex_chromatic_number}. These two voters prove that $M$ is $2$-rationalizable.

\begin{figure}[h]
\begin{minipage}{0.35\textwidth}
$$M = \begin{bmatrix}
    0 & 1 & 0 & \frac{1}{2} & 1 \\
    0 & 0 & \frac{1}{2} & \frac{1}{2} & \frac{1}{2} \\
    1 & \frac{1}{2}  & 0 & \frac{1}{2} & \frac{1}{2} \\
    \frac{1}{2} & \frac{1}{2} & \frac{1}{2} & 0 & 1 \\
    0 & \frac{1}{2} & \frac{1}{2} & 0 & 0
\end{bmatrix}$$
\end{minipage}
\hspace{0.5cm}
\begin{minipage}{0.3\textwidth}
\scalebox{0.8}{
    \begin{tikzpicture}[>=stealth',shorten >=1pt,auto,node distance=2.5cm,thick,
                    red node/.style={circle,draw=red,fill=red!20,font=\sffamily\Large\bfseries},
                    green node/.style={circle,draw=green,fill=green!20,font=\sffamily\Large\bfseries}]

      \node[red node] (1) {$1$};
      \node[green node] (2) [above right=0.8cm of 1] {$2$};
      \node[green node] (3) [below right=0.8cm of 2] {$3$};
      \node[red node] (4) [below=0.7cm of 1] {$4$};
      \node[green node] (5) [below=0.7cm of 3] {$5$};
      \path
    (1) edge (2)
    (1) edge (5)
    (3) edge (1)
    (4) edge (5);
    \end{tikzpicture}}
\end{minipage} 
\begin{minipage}{0.25\textwidth}
\scalebox{0.85}{
    \begin{tikzpicture}[->,>=stealth',shorten >=1pt,auto,node distance=2.5cm,thick,
                    blank node/.style={font=\sffamily\Large\bfseries}]
        
      \node[blank node] (2) {$2$};
      \node[blank node] (3) [below=0.6cm of 2] {$3$};
      \node[blank node] (5) [below=0.6cm of 3] {$5$};
      \node[blank node] (1) [right=0.6cm of 3] {$1$};
      \node[blank node] (4) [below=0.6cm of 1] {$4$};
      \node[red] () at (0.75,0.75) {Voter 1};
        \draw[rounded corners, red] (-0.4, 0.5) rectangle (1.5, -2.95) {};
      \path
    (2) edge (3)
    (3) edge (5)
    (1) edge (4);
    \end{tikzpicture}
    \quad
    \begin{tikzpicture}[->,>=stealth',shorten >=1pt,auto,node distance=2.5cm,thick,
                    blank node/.style={font=\sffamily\Large\bfseries}]
        
      \node[blank node] (5) {$5$};
      \node[blank node] (3) [below=0.6cm of 5] {$3$};
      \node[blank node] (2) [below=0.6cm of 3] {$2$};
      \node[blank node] (4) [right=0.6cm of 3] {$4$};
      \node[blank node] (1) [below=0.6cm of 4] {$1$};
           \node[blue] () at (0.75,0.75) {Voter 2};
      \draw[rounded corners, blue] (-0.4, 0.5) rectangle (1.5, -2.95) {};
      \path
    (5) edge (3)
    (3) edge (2)
    (4) edge (1);
    \end{tikzpicture}}
\end{minipage}
    \caption{A $2$-chromatic unanimity graph inducing two $2$-rational voters consistent with its half-integral preference matrix $M$.}
    \label{ex_chromatic_number}
\end{figure}
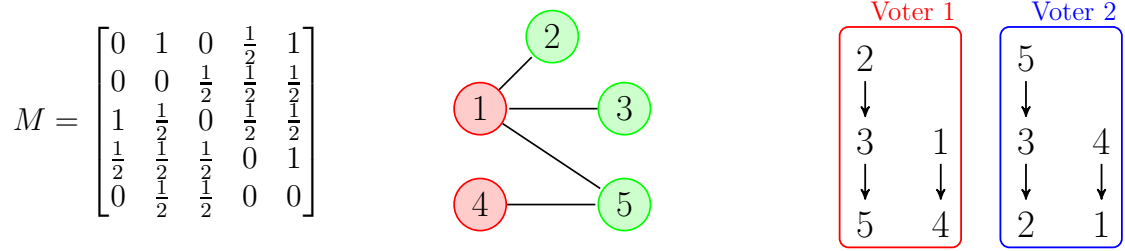

\subsection{A Lower Bound on the Rationality Number of Half-Integral Matrices}

Let's now show that the upper bound in Lemma~\ref{lem:upper-half} is tight (to within a constant factor) for the class
of half-integral preference matrices.

Take any integral preference matrix $M$. Its voting graph is a {\em tournament}. If this tournament is acyclic then the matrix is rationalizable, that is, it is $\alpha(M)=1$. But its unanimity graph is a clique so has chromatic number $\c(G_M)=n$.

Because an integral matrix is trivially half-integral, this example shows that is {\bf not} the case that
$\alpha(M)=\Omega(\c(G_M))$ for every half-integral matrix.
However, the upper bound is indeed tight in the following sense:
for any $k\in \mathbb{N}$, there exists a half-integral matrix $M$ whose unanimity graph has chromatic number $\c(G_M)=k$
and whose rationality number is at least $\Omega(k)$.

\begin{lemma}\label{lem:lower-half}
Given $k\in \mathbb{N}$, there exists a half-integral matrix $M$ with $\c(G_M)=k$ and $\alpha(M)\ge \frac15\cdot k$.
\end{lemma}
\begin{proof}
 Take any $k \in \N$. 
 We build a unanimity graph $G_M$ on $n$ vertices composed of
$k$ disjoint independent sets $\{C_1, C_2,\dots, C_k\}$, where
 $|C_\ell|=\frac{n}{k}$, for all $1\le \ell\le k$.
(Note we will determine below how large $n$ needs to be.) 
In $G_M$ every pair of candidates belonging to different independent sets is connected by an edge. Because any set of vertices from distinct independent sets forms a clique, it follows that $\c(G_M)=k$.

We now need to construct a voting graph $D_M$ and a half-integral preference matrix $M$ corresponding to this unanimity graph $G_M$. To do this we apply the probabilistic method.
We uniformly at random orient each edge in $G_M$ independently to produce $D_M$ (and thus $M$).

We prove the following key property. With non-zero probability, every subset $S \subseteq V(D_M)$ of size $5\frac{n}{k}$ contains a directed triangle. 
Let $S$ be an arbitrary subset of the vertices of size $5\frac{n}{k}$ and let $S_\ell = S\cap C_\ell$, for all $1\le \ell\le k$. Set $m_\ell = |S_\ell|$, the number of vertices of $S$ contained in $C_\ell$.
We desire a lower bound on the number of edge-disjoint triangles contained within $S$. 
To do this we apply the following ``merge'' operation.
Given $S_\ell$ and $S_\gamma$, where $\ell\neq \gamma$,
we set create an independent set $S_\ell\cup S_\gamma$
by removing the arcs between vertices in $S_\ell$ and $S_\gamma$.
Clearly, by removing arcs we cannot increase the number of triangles. So any lower bound we obtain after applying this operation applies to the original instance.

Our goal is to obtain (at least) two independent sets of size at least $\frac{n}{k}$. So if $m_1 < \frac{n}{k}$, we merge other sets into $S_1$
until $m_1\ge \frac{n}{k}$. 
Observe that at this point we must have $m_1< 2\cdot \frac{n}{k}$. We then repeat this process on another set, say $S_2$, until $\frac{n}{k} \le m_2 < 2\cdot \frac{n}{k}$. 
There are now at least $5\frac{n}{k} - m_1 - m_2 \geq \frac{n}{k}$ vertices remaining. So $m_3+\cdots+ m_k\ge \frac{n}{k}$. 
    
We can now show that there are at least $\left( \frac{n}{k} \right)^2$ edge-disjoint triangles in $S$. 
First observe that every vertex in $S_1$ is adjacent to every vertex in $S_2$. 
By deleting vertices we may assume $m_1=m_2=\frac{n}{k}$.
Then, by repeatedly applying Hall's theorem, 
we can find $\frac{n}{k}$ edge-disjoint perfect matchings between $S_1$ and $S_2$.
Note that each of these perfect matchings has cardinality 
$\frac{n}{k}$.

Next, by deleting vertices, we may assume $\bar{S}=S_3\cup\cdots S_k$ contains exactly $\frac{n}{k}$ vertices, that is, 
$\sum\limits_{i \geq 3} m_i = \frac{n}{k}$. 
Now pair each perfect matching with a different vertex $v$ from $\bar{S}$. Since $v$ is adjacent in $G_M$ to each vertex in $S_1\cup S_2$ this creates $\frac{n}{k}$ disjoint triangles.
Because there are $\frac{n}{k}$ perfect matchings paired to  
$\frac{n}{k}$ distinct vertices in $\bar{S}$, this gives a total
of $(\frac{n}{k})^2$ triangles in total, as claimed.

But each of these triangles is a directed $3$-cycle with probability $\frac{1}{4}$. Furthermore, as each of these triangles are edge-disjoint, these are independent events. So the probability that $S$ contains no directed cycle, which is less than the probability that none of our triangles are directed, is upper bounded as follows.
$$\mathbb{P}\{ S \text{ has no directed cycle} \} 
\leq \mathbb{P}\{\text{none of the triangles are directed}\} 
= \left( \frac{3}{4} \right)^{\left(\frac{n}{k}\right)^2}$$
    
Thus, by the union bound, we can show that the probability of the existence of such a subset $S$ without directed cycle is bounded away from $1$.  Specifically
    \begin{align*}
        \mathbb{P}\{\exists S \text{ with no directed cycle}\} 
        \leq \sum_{S} \mathbb{P}\{ S \text{ has no directed cycle} \}
        \leq N \times \left( \frac{3}{4} \right)^{\left(\frac{n}{k}\right)^2}
    \end{align*}
    where $N$ is the number of subsets of $V(G_M)$ of size $5 \cdot \frac{n}{k}$. Thus
    \begin{align*}
        N = \binom{n}{\frac{5n}{k}} 
        = \frac{n(n-1) \cdots (n - \frac{5n}{k} +1)}{\left(\frac{5n}{k}\right)!}
        \leq \frac{n^{\frac{5n}{k}}}{\left(\frac{5n}{ke}\right)^{\frac{5n}{k}}}
        = \left(\frac{ke}{5}\right)^{\frac{5n}{k}}
    \end{align*}
    where in the denominator we used the fact that $t! \geq \left(\frac{t}{e}\right)^t$ for any integer $t$, by Stirling's formula. Hence:
    \begin{align*}
        N \times \left( \frac{3}{4} \right)^{\left(\frac{n}{k}\right)^2} 
        \leq \left(\frac{ke}{5}\right)^{\frac{5n}{k}} \cdot \left( \frac{3}{4} \right)^{\left(\frac{n}{k}\right)^2} 
    \end{align*}
    which is strictly less than 1 for $n$ large enough. To see this, raising this expression to the power $\frac{k}{n}$ we have $\left(\frac{ke}{5}\right)^{5} \cdot \left( \frac{3}{4} \right)^{\frac{n}{k}} $
    which goes to zero as $n \rightarrow \infty$. Thus, we conclude that 
    \begin{align*}
        \mathbb{P}\{ \text{every subset } S \text{ has a cycle}\} 
        = 1 - \mathbb{P}\{\text{there exists subset }S \text{ with no cycle}\} > 0.
    \end{align*}
Therefore, by the probabilistic method, there exists an orientation of $G_M$ for which every subset $S \subseteq V(G_M)$ of size $5\cdot\frac{n}{k}$ contains a directed triangle. This implies that the longest chain any voter can have in its partial order is of length at most $5\frac{n}{k}$.
Consequently each voter must have at least
$\frac{k}{5}$ chains in their partial order. This proves that this orientation $D_M$ of $G_M$ requires that each voter be 
at best $\frac{k}{5}$-rational. Hence, for this voting graph $D_M$, we have $\alpha(M)\ge \frac{k}{5}$.
\end{proof}

Together Lemma~\ref{lem:upper-half} and Lemma~\ref{lem:lower-half} prove our first main result, Theorem~\ref{thm:half-integral}. Observe the lower and upper bounds are tight 
up to a constant factor of $5$. It is an intriguing combinatorial problem to completely close the gap between the lower and upper bounds.

\section{Integral Preference Matrices}\label{sec:integral}

Stronger results can be obtained when the preference matrix $M$ is integral, that is, $p_{ij}\in \{0,1\}$ for all $i,j$.

\subsection{One Voter Suffices}
In general, if $M$ is $\alpha$-rational more than one voter
might be required to $\alpha$-rationalize the matrix.
However, if $M$ is integral, then a single voter
can $\alpha$-rationalize it.

\begin{theorem}\label{thm:one-voter}
Let $M$ be an integral preference matrix. If $M$ is $\alpha$-rational then it is consistent with a single
$\alpha$-rational voter.
\end{theorem}
\begin{proof}
Take an integral preference matrix $M$. Recall that its 
voting graph is then a tournament. Now, take any pair of candidates, $i$ and $j$. Without loss of generality, $p_{ij}=1$.
The rationality constraints (\ref{Requirements}) then imply that the
fraction of voters that weakly prefer $i$ over $j$ must be one.
That is, {\em every} voter must either prefer $i$ over $j$ or
must be indifferent between $i$ and $j$.

In particular, if $i\succ_v j$ in the partial order of voter $v$ then it must be the case that $p_{ij}=1$. If $p_{ij}=0$ then (\ref{Requirements}) is violated. Thus any strict preference in the partial order $\succ_v$ must agree with the
preference matrix $M$. Any indifference in the the 
partial order $\succ_v$ imposes no constraint on $p_{ij}$.

Now assume $M$ is consistent with a collection of
$m$ $\alpha$-rational voters. By the above argument,
each of these $m$ voters must satisfy (\ref{Requirements})
on its own! So each such voter suffices to 
$\alpha$-rationalize $M$.
Therefore, if $M$ is $\alpha$-rational then it is consistent with a single $\alpha$-rational voter. 
\end{proof}

\subsection{The Dichromatic Number}

Using Theorem~\ref{thm:one-voter}, we can now present a characterization of the rationality
number of an integral preference matrix $M$ in terms of a 
colouring of its voting graph $D_M$.
Now $D_M$ is a directed graph, so what do we mean by a colouring of a directed graph?
A {\em chromatic colouring} of an undirected graph $G$ is a partition of the vertices into independent sets.
A {\em dichromatic colouring} of a directed graph $D$ is a partition of the vertices into acyclic sets.
A dichromatic colouring can be viewed as a generalisation of a chromatic colouring of an undirected graph.
To see this, if we bidirect every edge in an undirected graph $G$ then an acyclic set in the resultant directed graph $D$ is an independent set in the original undirected graph.
Analogous to the chromatic number of an undirected graph, Neumann-Lara~\cite{NL82} defined the \textit{dichromatic number} $\dichi(D)$ of a digraph $D$ to be the
minimum number of colours required in any dichromatic colouring of $D$.
It is particularly important in the study of the voting graph $D_M$ of an integral preference matrix, as the following result shows:

\begin{theorem}\label{thm:dichromatic}
Let $M$ be an integral preference matrix.
Then its rationality number is equal to the dichromatic number
of its voting graph: $\alpha(M)= \dichi(D_M)$.
\end{theorem}
\begin{proof} 
Let $M$ be an integral preference matrix.
Then its voting graph $D_M$ is a tournament.  
First assume that $\dichi(D_M)=k$.
Then we can partition the vertices of $D_M$ into $k$ acyclic
subgraphs $\{C_1,C_2,\dots, C_k\}$. Take any $C_\ell$, for $1\le \ell \le k$. Then $C_\ell$ is itself an tournament.
Because it it is acyclic it has an acyclic ordering. Furthermore, this ordering is unique as $C_\ell$ is a tournament. We use acyclic ordering as total order to induce
a chain on $C_\ell$. In this way we have a partial order that consists of $k$ disjoint chains on $\{C_1,C_2,\dots, C_k\}$.
This partial order corresponds to a single $\alpha$-rational voter. Moreover if $i\succ j$ in this partial order then $p_{ij}=1$. Thus by the argument of Theorem~\ref{thm:one-voter}
this single voter is consistent with $M$ in satisfying the
constraints (\ref{Requirements}). Thus $\alpha(M)\le k$.

Second assume that $\alpha(M)=k$.
Then, by Theorem~\ref{thm:one-voter}, there is a single $\alpha$-rational voter $v$ that is consistent with $M$.
Let $\succ$ be the partial order of voter $v$.
Then, by Corollary~\ref{cor:chains}, we can assume the partial
order $\succ$ consists of exactly $k$ disjoint chains,
$\{C_1,C_2,\dots, C_k\}$.
We claim that each chain induces an acyclic subgraph in
the voting graph $D_M$.
Suppose not, then there exist candidates $i$ and $j$ such that
$i\succ j$ and $p_{ij}=0$. But this contradicts (\ref{Requirements}) since the fraction of voters that strongly prefer $i$ over $j$ is one. So voter $v$ is not consistent with $M$, a contradiction.
So $\{C_1,C_2,\dots, C_k\}$ are a partition of the vertices in $D_M$ into acyclic subgraphs. Thus $\dichi(D_M)\le k$.

Putting this together we have $\alpha(M)=\dichi(D_M)$ as desired.
\end{proof}

\noindent{\tt Example VI.}
Consider the integral preference matrix $M$ shown in Figure \ref{ex_equiv_char}. Its voter graph has dichromatic number $3$
and as illustrated is consistent with a single $3$-rational voter.
\vspace{-0.3cm}
\begin{figure}[h]
\begin{minipage}{0.4\textwidth}
$$M = \begin{bmatrix}
    0 & 1 & 0 & 1 & 0 & 1 & 0 \\
    0 & 0 & 1 & 0 & 1 & 1 & 0 \\
    1 & 0 & 0 & 1 & 0 & 0 & 1 \\
    0 & 1 & 0 & 0 & 1 & 0 & 1 \\
    1 & 0 & 1 & 0 & 0 & 1 & 0 \\
    0 & 0 & 1 & 1 & 0 & 0 & 1 \\
    1 & 1 & 0 & 0 & 1 & 0 & 0
\end{bmatrix}$$
\end{minipage}
\hspace{0.5cm}
\begin{minipage}{0.3\textwidth}
\scalebox{0.8}{
    \begin{tikzpicture}[->,>=stealth',shorten >=1pt,auto,node distance=2.5cm,thick,
        blue node/.style={circle,draw=blue,fill=blue!20,font=\sffamily\Large\bfseries},
        red node/.style={circle,draw=red,fill=red!20,font=\sffamily\Large\bfseries},
        green node/.style={circle,draw=green,fill=green!20,font=\sffamily\Large\bfseries}]

      \node[blue node] (1) at (90:1.6) {$1$};
      \node[green node] (2) at (38.6:1.6) {$2$};
      \node[blue node] (3) at (-12.9:1.6) {$3$};
      \node[blue node] (4) at (-64.3:1.6) {$4$};
      \node[red node] (5) at (-115.7:1.6) {$5$};
      \node[red node] (6) at (-167.1:1.6) {$6$};
      \node[green node] (7) at (-218.6:1.6) {$7$};
      \path
    (1) edge (2)
    (1) edge (4)
    (1) edge (6)
    (3) edge (1)
    (5) edge (1)
    (7) edge (1)
    (2) edge (3)
    (2) edge (5)
    (2) edge (6)
    (3) edge (4)
    (3) edge (7)
    (4) edge (2)
    (4) edge (5)
    (4) edge (7)
    (5) edge (3)
    (5) edge (6)
    (6) edge (3)
    (6) edge (4)
    (6) edge (7)
    (7) edge (2)
    (7) edge (5);
    \end{tikzpicture}}
\end{minipage} 
\begin{minipage}{0.2\textwidth}
\scalebox{0.85}{
    \begin{tikzpicture}[->,>=stealth',shorten >=1pt,auto,node distance=2.3cm,thick,
                    blank node/.style={font=\sffamily\Large\bfseries}]

      \node[blank node] (3) {$3$};
      \node[blank node] (1) [below=0.5cm of 3] {$1$};
      \node[blank node] (4) [below=0.5cm of 1] {$4$};
      \node[blank node] (5) [right=0.5cm of 1] {$5$};
      \node[blank node] (6) [below=0.5cm of 5] {$6$};
      \node[blank node] (7) [right=0.5cm of 5] {$7$};
      \node[blank node] (2) [below=0.5cm of 7] {$2$};
      \node[red] () at (1.5,0.75) {Voter 1};

    \draw[rounded corners, red] (-0.4, 0.5) rectangle (2.6, -2.75) {};
        
    \path
    (3) edge (1)
    (1) edge (4)
    (7) edge (2)
    (5) edge (6);
    \end{tikzpicture}}
\end{minipage}
    \caption{An integral preference matrix $M$, its $3$-dichromatic voting graph $D_M$, and a corresponding $3$-rational voter.}
    \label{ex_equiv_char}
\end{figure}
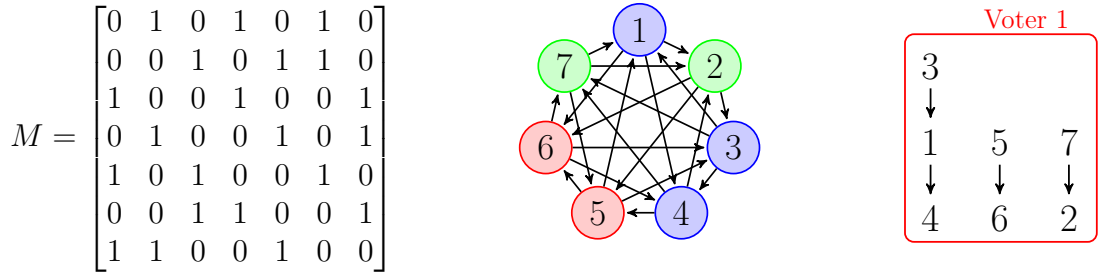
\vspace{-0.3cm}

\subsection{An Upper Bound on the Rationality Number of Integral Matrices}


For the class $\mathcal{M}^{0/1}$ of integral preference matrices we can strengthen
the bounds in Theorem~\ref{thm:half-integral}.
\begin{lemma}\label{lem:upper-integral}
If $M$ is an integral preference matrix then
$\alpha(M) \le \frac{3n}{\log n} = \frac{3\c(G_M)}{\log n}$.
\end{lemma}
\begin{proof}
By Theorem~\ref{thm:dichromatic}, to analyze the rationality number of an integral preference matrix $M$ we must
consider dichromatic colorings of the voting graph $D_M$.
In particular, any algorithm to dichromatic colour the voting graph will give an upper bound on the rationality number of $M$.
We now present a greedy algorithm which gives tight bounds
on the rationality number over the class of integral preference matrices.

The greedy algorithm selects a colour class $C_1$ as follows.
It picks the vertex $v_1$ with the highest out-degree in $D_M$
and adds it to $C_1$. It then selects the vertex $v_2$ with the
highest out-degree in the subgraph $V(D_M)\cap \Gamma^+(v_1)$, induced by the set of out-neighbours of $v_1$, and adds $v_2$
to $C_1$. It then selects the vertex $v_3$ with the
highest out-degree in the subgraph $V(D_M)\cap \Gamma^+(v_1)\cap \Gamma^+(v_2)$, induced by the set of out-neighbours of both 
$v_1$ and $v_2$, and adds $v_3$ to $C_1$.
This process terminates with $C_1=\{v_1,v_2,\dots, v_t\}$
when $\bigcap_{i=1}^t \Gamma^+(v_i) \cap V(D_M) =\emptyset$.
The vertices in $C_1$ are given the colour $1$. The algorithm
is then repeated on $D_M\setminus C_1$ to select the second colour class $C_2$, etc. The entire process terminates
when every vertex has been coloured. 

Let this greedy algorithm output the colour classes $\{C_1, C_2, \dots, C_k\}$. We claim this is a valid dichromatic colouring of $D_m$. This is true because each $C_\ell$, for $1\le \ell\le k$,
is acyclic. In particular, by construction, if the
$C_\ell=\{v_1,v_2,\dots, v_r\}$, then $v_j$ is an out-neighbour of
$v_i$, for any $i<j$. Thus $C_\ell$ contains no directed cycle.
Furthermore, we claim that $k \le \frac{3n}{\log n}$.
That is, the greedy algorithm gives a $k$-dichromatic colouring 
of $D_M$ using at most $\frac{3n}{\log n}$, where $n$
is the number of vertices (candidates).
Since, $\c(G_M)=n$, by Theorem~\ref{thm:dichromatic}, this will
prove that $\alpha(M) \le \frac{3n}{\log n} = \frac{3\c(G_M)}{\log n}$. So let's verify this claim.

To do this, observe that if there are at least $\frac{n}{2^i}$ vertices remaining when the algorithm begins to construct a new colour class $C_\ell$, then the cardinality of $C_\ell$ will be at least $\log\frac{n}{2^i}= \log n -i$. Indeed, as the graph is a tournament, at any time there is a vertex whose out-degree is at least (the floor of) half the number of vertices under consideration. 
In particular, consider the first time we have at most $\frac{n}{2^{i-1}}$ vertices remaining. Then the number of colour classes we find until the number of remaining vertices is at most 
$\frac{n}{2^{i}}$ is upper bounded by $\frac{\frac{n}{2^i}}{\log n -i}$.

If there are less than $\frac{n}{\log n}$ vertices remaining then the number of colour classes the greedy algorithm finds from that point on is trivially upper bounded by $\frac{n}{\log n}$.
Thus the total number of colour classes the greedy algorithm finds in colouring every vertex is at most
\begin{eqnarray*}
\frac{n}{\log n} + \sum_{i=1}^{\log\log n} \, \frac{\frac{n}{2^i}}{\log n -i} 
&=& \frac{n}{\log n} + n\cdot\sum_{i=1}^{\log\log n} \, \frac{1}{2^i\cdot (\log n -i)} \\
&\le& \frac{n}{\log n} + n\cdot\sum_{i=1}^{\log\log n} \, \frac{1}{2^i\cdot (\log n -\log\log n)}\\
&\le& \frac{n}{\log n} + n\cdot\sum_{i=1}^{\log\log n} \, \frac{1}{2^{i-1}\cdot \log n}\\
&\le& \frac{3n}{\log n}
\end{eqnarray*}
This gives our upper bound on the rationality number $\alpha(M)$.
\end{proof}

\subsection{A Lower Bound on the Rationality Number of Integral Matrices}

Let's now show that the upper bound in Lemma~\ref{lem:upper-integral} is tight (to within a constant factor) for the class
of integral preference matrices.

\begin{lemma}\label{lem:lower-integral}
There exists an integral matrix $M$ with 
$\alpha(M)\ge \frac{n}{2\log n +1} = \frac{\c(G_M)}{2\log n+1}$.
\end{lemma}
\begin{proof}

We claim there is a tournament with dichromatic number at least $\frac{n}{2\log n+1}$. To prove this we again apply the probabilistic method. Take a random tournament on $n$ vertices. 
Next select any subset $S$ of $k$ vertices. 
There are $k!$ acyclic orderings of a tournament on $k$ vertices
and $2^{k\choose 2}$ ways to orient the arcs in the tournament.
Thus, the probability that $S$ is a acyclic is exactly $\frac{k!}{2^{k\choose 2}}$. Furthermore
there are ${n\choose k}$ ways to choose $S$. So, by the union bound, the probability that at least one of them induces an acyclic tournament is at most
\begin{equation*}
{n\choose k}\cdot \frac{k!}{2^{k\choose 2}}
\ <\ \frac{n^k}{2^{k\choose 2}} 
\ =\ \frac{n^k}{2^{\frac12 (k-1)k}} 
\ =\ \left(\frac{n}{2^{\frac12 (k-1)}}\right)^k 
\end{equation*}
But this is less than $1$ if $2^{\frac12 (k-1)}\ge n$. That is, if $k\ge 2\log n +1$.
This implies there exists a tournament $D_M$ on $n$ vertices that contains no acyclic subgraphs of cardinality greater than $2\log n +1$. For this tournament, every dichromatic colour class has cardinality at most $2\log n +1$. Thus, its dichromatic number is
at least $\frac{n}{2\log n+1}$.
Now $D_M$ is the voting graph of an integral preference matrix $M$. So, by Theorem~\ref{thm:dichromatic}, the rationality number of $M$ is $\alpha(M)=\dichi(D_M)\ge \frac{n}{2\log n+1}$.
The result follows.
\end{proof}

Together Lemma~\ref{lem:upper-half} and 
Lemma~\ref{lem:lower-half} prove our second main result, Theorem~\ref{thm:integral}. Again, closing the factor $6$
gap between the lower and upper bounds is an interesting open problem.

\section{Computational Complexity}\label{sec:complexity}

\begin{theorem}\label{thm:hardness}
The rationality number problem is NP-complete for $k\ge 2$, even for the case of integral preferences matrices. 
\end{theorem}
\begin{proof}
Given an integral preference matrix $M$.
By Theorem~\ref{thm:dichromatic}, determining whether $\alpha(G)\le k$
is equivalent to deciding whether the tournament $D_M$ has dichromatic number $k$.
Consider first the case of $k=2$. Now a tournament $T$ has dichromatic number $2$ if and only if the vertices of $T$ can be partitioned into two feedback vertex sets. 
To see this, observe that the complement of a feedback vertex set induces, by definition, an 
acyclic graph. Using a reduction from \textsc{Not-All-Equal-3SAT}, Chen, Hu and Zang~\cite{CHZ07} proved that determining if a tournament can be partitioned
into two feedback vertex sets is NP-complete. 

Next consider $k\ge 3$. Fox et al.~\cite{Fox17} 
gave an alternate proof, via a reduction from the \textsc{Triangle-Free-Cut} problem, that deciding whether the tournament has dichromatic number $2$ is NP-complete. Moreover, they give a reduction from 
$(k-1)$-dicolorability to $k$-dicolorability, hence proving NP-completeness for all $k\geq 2$. The reduction is simple. Given a tournament $T$, construct a new tournament $\hat{T}$ consisting of two identical copies of $T$ and an extra vertex $z$, connected in the order $T_1 \rightarrow T_2 \rightarrow z \rightarrow T_1$. It can easily be verified that $\dichi(T)=k-1$ if and only if $\dichi(\hat{T})=k$. The theorem follows. 
\end{proof}

This hardness result indicates that it may be fruitful to search for approximation algorithms for the rationality number of a preference matrix. We remark that for the special case of integral
preference matrices with $\alpha(M)=2$ a $5$-approximation can be
derived from the work of Klingelhoefer and Newman~\cite{KN23}.
Specifically, they showed how to $10$-dicolor in polynomial time a $2$-dicolorable tournament.

\section{Acknowledgements}
We are grateful to Sophie Spirkl for showing us an elegant reduction from {\sc Monotone-Not-All-Equal-3-SAT} to the problem of deciding whether a tournament has dichromatic number 2 and to Gerardo Berbeglia for interesting discussions.

\bibliographystyle{splncs04}
\bibliography{references}

\end{document}